\documentclass[11pt,a4paper,final]{article}
\usepackage{amsmath}%
\numberwithin{equation}{section}
\usepackage{amstext}%
\usepackage{amssymb}%
\usepackage{showkeys}%
\usepackage{epsfig}%
\usepackage{graphicx}%
\usepackage{xcolor}
\usepackage{multicol}
\usepackage{amsthm}
\usepackage{booktabs}
\usepackage{apacite}
\usepackage{multirow}
\usepackage[top=1in, bottom=1.25in, left=0.8in, right=0.8in]{geometry}

\theoremstyle{prop}
\theoremstyle{proof}

\newtheorem{prop}{Proposition}
\newtheorem{theorem}{Theorem}
\newtheorem{cor}{Corollary}
\newtheorem{lemma}[theorem]{Lemma}
\newtheorem{ass}{Assumption}

\begin{document}
\title{Derivatives of Risk Measures}
\author{Battulga Gankhuu\footnote{Department of Applied Mathematics, National University of Mongolia; E-mail: battulgag@num.edu.mn; Phone Number: 976--99246036}}
\date{}

\maketitle 

\begin{abstract}
This paper provides the first and second order derivatives of any risk measures, including VaR and ES for continuous and discrete portfolio loss random variable variables. Also, we give asymptotic results of the first and second order conditional moments for heavy--tailed portfolio loss random variable. 
\end{abstract}



\section{Introduction}

Derivatives of risk measures are useful for capital allocation, sensitivity analysis, and portfolio choice problems. Capital allocation means that a risk measure of some investment is represented as a sum over the quantities, which are related to each individual risk source. Understanding of the risk sources can be applied in performance measurement, strategic planning, pricing, and communication with regulators, rating agencies, and security analysts. Overviews of a variety of different methodologies for capital allocation are given in \citeA{Venter06} and \citeA{Guo21}.

The most popular and simple capital allocation principle is the Euler (gradient) allocation principle. The Euler allocation principle requires that a differentiable risk measure must be a positive homogeneous function. In this case, the risk measure is represented in terms of partial derivatives of the risk measure. For formal analysis of the problem of capital allocation, which is based on the Euler allocation principle of a risk measure, appeared in the banking and insurance industry, we refer to  \citeA{Tasche99}, \citeA{Tasche04}, and \citeA{McNeil15}, among others. 

Another important application of the partial derivative of a risk measure is the sensitivity analysis. For example, the following question arises: if we add a new position in a portfolio or if we increase (decrease) a particular position in the portfolio, how does the risk of the entire portfolio change? This question can be answered using the partial derivative of a risk measure.

Also, the partial derivative of a risk measure is useful for the portfolio choice problem. For the portfolio choice problem, if the objective function is the risk measure and a constraint set is convex, then to check the problem has a unique solution, we must check the strict convexity of the risk measure. Consequently, one needs the second order partial derivatives of the risk measure.

The most popular risk measures in practice are Value--at--Risk (VaR) and Expected Shortfall (ES), which are two closely related. Roughly speaking, VaR is a maximum loss of the portfolio loss random variable at a certain confidence level and ES is a mean of the portfolio loss random variable given that the portfolio random variable is greater than or equal to the VaR at the same confidence level. 

The first order partial derivatives of VaR go back to \citeA{Tasche99,Tasche00} and \citeA{Gourieroux00}. They have shown that the first order derivative of VaR equals a conditional expectation of the first order derivative of the portfolio loss random variable given that the portfolio loss random variable equal to the VaR. By using the same method, \citeA{Tasche99,Tasche00} also has shown that the first order derivative of ES equals a conditional expectation of the first order derivative of portfolio loss random variable given that the portfolio loss random variable greater than or equal to the VaR. The second order derivatives of VaR were calculated by \citeA{Gourieroux00}. The all results of \citeA{Tasche99,Tasche00} and \citeA{Gourieroux00} are relied on continuous portfolio loss random variable. While \citeA{Rau-Bredow03} has introduced the first and second order derivatives of VaR and ES based on discrete portfolio loss random variable. Also, \citeA{Glasserman05} has provided simulation methods that calculate the first order derivatives of VaR and ES for credit portfolio. In this paper, we obtain the first and second order derivatives of any risk measures for both continuous and discrete portfolio loss random variable variables. 

This paper is organized as follows: In Section 2, we give the first and second order derivatives of any risk measures, including VaR and ES for discrete portfolio loss random variable. Section 3 is dedicated to the first and second order derivatives of any risk measures, including VaR and ES for continuous portfolio loss random variable. In Section 4, we provide asymptotic results of the first and second order conditional moments for heavy--tailed portfolio loss random variable. Finally, in Section 4, we summarize the results.

\section{Discrete Portfolio Loss Random Variable}

Let $L:\Omega\times\mathbb{R}^d\to \mathbb{R}$ be a portfolio loss random variable, defined on probability space $(\Omega,\mathcal{F},\mathbb{P})$ and $\varrho: \mathcal{M}\to\mathbb{R}$ be a risk measure, where $\mathcal{M}\subset \mathcal{L}^2(\Omega,\mathcal{F},\mathbb{P})$ is linear space and $\mathcal{L}^2(\Omega,\mathcal{F},\mathbb{P})$ is the space of square integrable random variables. We suppose that there are $d$ assets and for each asset $i=1,\dots,d$, let $L_i$ be a loss random variable of $i$--th asset. Then, a portfolio loss is given by the following linear equation
\begin{equation}\label{2.01}
L(x)=L_1x_1+\dots+L_nx_n,
\end{equation}
where the deterministic variable $x_i$ is a weight of $i$--th asset. In this paper, we assume that the risk measure $\varrho(x):=\varrho(L(x))$ is a continuous function and has second order derivatives with respect to argument $x=(x_1,\dots,x_d)'\in\mathbb{R}^d$. Which means the following Assumption holds
\begin{ass}\label{a.01}
Risk measure $\varrho(x)$ is a continuous function and has second order partial derivatives.
\end{ass}
To obtain partial derivatives of the risk measure $\varrho(x)$ with respect to its individual arguments, we introduce the following notations: $H(x):=L(x)-\varrho(x)$ is the difference between portfolio loss random variable $L(x)$ and the risk measure $\varrho(x)$, for $j=1,\dots,d$, $\tilde{L}_j:=(L_1,\dots,L_{j-1},L_{j+1},\dots,L_d)'\in\mathbb{R}^{[d-1]\times 1}$ is a random vector, which is composed of $L_i$s except $L_j$, and $\tilde{L}_j(x):=\sum_{i\neq j}x_iL_i=L(x)-x_jL_j$ is a difference between the portfolio loss random variable $L(x)$ and a random variable $x_jL_j$. To keep calculations convenient, we assume the following Assumption holds throughout the paper:
\begin{ass}\label{a.02}
Integral and derivative operators can be interchanged.
\end{ass}
For the necessary conditions of the integral and derivative operators to be interchanged in their order, we refer to \citeA{Tasche99}, \citeA{Tasche00}, and \citeA{Klenke13}.

In this Section, we develop partial derivative formulas of the risk measure for discrete loss random vector $(L_1,\dots,L_d)'$. In order to obtain partial derivatives of the risk measure, corresponding to the discrete portfolio loss random variable $L(x)$, we need first order partial derivatives of probabilities $\mathbb{P}[H(x)=0|\tilde{L}_j]$ and $\mathbb{P}[H(x)=0]$ and an expectation $\mathbb{E}\big[L_j1_{\{H(x)\geq 0\}}\big]$ for $j=1,\dots,d$. In the following Lemma, we give the partial derivatives of the probabilities and expectation.

\begin{lemma}\label{l.01}
Let $(L_1,\dots,L_d)'$ be an integrable discrete random vector. Then, for a sufficiently small $\varepsilon>0$ and $i,j=1,\dots,d$, it holds
\begin{equation}\label{2.02}
\frac{\partial}{\partial x_i}\mathbb{P}\big[|H(x)|\leq \varepsilon\big|\tilde{L}_j\big]=0,
\end{equation}
\begin{equation}\label{2.03}
\frac{\partial}{\partial x_i}\mathbb{P}\big[|H(x)|\leq \varepsilon\big]=0,
\end{equation}
\begin{equation}\label{2.04}
\frac{\partial}{\partial x_i}\mathbb{E}\big[L_j1_{\{H(x)\geq 0\}}\big]=0,
\end{equation}
and
\begin{equation}\label{2.05}
\frac{\partial}{\partial x_i}\mathbb{E}\big[L(x)1_{\{H(x)\geq 0\}}\big]=\mathbb{E}\big[L_i1_{\{H(x)\geq 0\}}\big],
\end{equation}
where for a generic event $A\in \mathcal{F}$, $1_A$ is an indicator random variable of the event $A$.
\end{lemma}

\begin{proof}
Let us assume that $k=1,2,\dots$, the random variable $H(x)$ takes values $h_k(x)$, including zero, namely, $\mathbb{P}\big[H(x)=h_k(x)\big|\tilde{L}_k\big]>0$. Also, we suppose that for each $i=1,\dots,d$, loss random variable $L_i$ takes values $\ell_{ik}$ for $k=1,2,\dots$.  We consider a difference quotient for a probability $\mathbb{P}\big[|H(x)-h_k|\leq \varepsilon\big|\tilde{L}_j\big]$, that is,
\begin{equation}\label{2.06}
m\bigg(\mathbb{P}\bigg[\bigg|H\big(x+\frac{1}{m}e_i\bigg)g-h_k\bigg(x+\frac{1}{m}e_i\bigg)\bigg|\leq \varepsilon\bigg|\tilde{L}_j\bigg]-\mathbb{P}\big[|H(x)-h_k|\leq\varepsilon\big|\tilde{L}_j\big]\bigg),
\end{equation}
where $e_i\in\mathbb{R}^d$ is a unit vector, whose $i$--th component equals 1 and others zero. Because the risk measure $\varrho(x)$ is a continuous function at the point $x$, there is a $\varepsilon>\delta>0$ and $m(\delta)>0$ such that for all $m>m(\delta)$,
\begin{equation}\label{2.07}
\bigg|\varrho\bigg(x+\frac{1}{m}e_i\bigg)-\varrho(x)-\frac{1}{m}L_i\bigg|<\delta/2
\end{equation}
and
\begin{equation}\label{2.07}
\bigg|\varrho\bigg(x+\frac{1}{m}e_i\bigg)-\varrho(x)-\frac{1}{m}\ell_{ik}\bigg|<\delta/2
\end{equation}
As a result, since $H\big(x+\frac{1}{m}e_i\big)=H(x)-\big\{\varrho\big(x+\frac{1}{m}e_i\big)-\varrho(x)-\frac{1}{m}L_i\big\}$, $h_k\big(x+\frac{1}{m}e_i\big)=h_k(x)-\big\{\varrho\big(x+\frac{1}{m}e_i\big)-\varrho(x)-\frac{1}{m}\ell_{ik}\big\}$, and 
\begin{equation}\label{ad.01}
\mathbb{P}\big[H(x)=h_k\big|\tilde{L}_j\big]=\mathbb{P}\big[|H(x)-h_k|\leq \varepsilon\big|\tilde{L}_j\big]=\mathbb{P}\big[|H(x)-h_k|\leq \varepsilon+\delta\big|\tilde{L}_j\big],
\end{equation}
for all $m>m(\delta)$, the difference quotients are equal to zero. Hence we have that 
\begin{equation}\label{ad.02}
\frac{\partial}{\partial x_i}\mathbb{P}\big[|H(x)-h_k|\leq \varepsilon\big|\tilde{L}_j\big]=0.
\end{equation}
That means equation \eqref{2.02} holds. For equation \eqref{2.03}, by the iterated expectation formula, we get that
\begin{equation}\label{2.08}
\frac{\partial}{\partial x_i}\mathbb{P}\big[|H(x)|\leq \varepsilon\big]=\mathbb{E}\bigg[\frac{\partial}{\partial x_i}\mathbb{P}\big[|H(x)|\leq \varepsilon\big|\tilde{L}_j\big]\bigg]=0.
\end{equation}
Thus, equation \eqref{2.03} is true. Due to the iterated expectation formula, for $j\neq k$, we have
\begin{equation}\label{2.09}
\frac{\partial}{\partial x_i}\mathbb{E}\big[L_j1_{\{H(x)\geq 0\}}\big]=\mathbb{E}\bigg[L_j\frac{\partial}{\partial x_i}\mathbb{P}[H(x)\geq 0|\tilde{L}_k]\bigg].
\end{equation}
Let $h_k^+\geq 0$ be nonnegative values of $H(x)$. Then, a probability $\mathbb{P}\big[H(x)\geq 0\big|\tilde{L}_k\big]$ is represented by
\begin{equation}\label{ad.03}
\mathbb{P}\big[H(x)\geq 0\big|\tilde{L}_k\big]=\sum_{k=1}^\infty\mathbb{P}\big[|H(x)-h_k^+|\leq \varepsilon\big|\tilde{L}_k\big].
\end{equation}
According to equation \eqref{ad.02}, we find that
\begin{equation}\label{2.10}
\frac{\partial}{\partial x_i}\mathbb{P}[H(x)\geq 0|\tilde{L}_k]=0.
\end{equation}
As a result, it follows from equations \eqref{2.09} and \eqref{2.10} that equation \eqref{2.04} holds. As
\begin{equation}\label{2.11}
\mathbb{E}\big[L(x)1_{\{H(x)\geq 0\}}\big]=\sum_{j=1}^d\mathbb{E}\big[L_j1_{\{H(x)\geq 0\}}\big]x_j,
\end{equation}
a partial derivative of the above equation with respect to the argument $x_i$ is obtained by
\begin{equation}\label{2.12}
\frac{\partial}{\partial x_i}\mathbb{E}\big[L(x)1_{\{H(x)\geq 0\}}\big]=\sum_{j=1}^d\frac{\partial}{\partial x_i}\mathbb{E}\big[L_j1_{\{H(x)\geq 0\}}\big]x_j+\mathbb{E}\big[L_i1_{\{H(x)\geq 0\}}\big].
\end{equation}
Consequently, from equations \eqref{2.04} and \eqref{2.12}, one gets equation \eqref{2.05}.
\end{proof}

Since $\mathbb{E}\big[H(x)1_{\{H(x)=0\}}\big]=0$, it holds
\begin{equation}\label{2.13}
\mathbb{E}\big[\varrho(x)1_{\{H(x)=0\}}\big]=\sum_{j=1}^d\mathbb{E}\big[L_j1_{\{H(x)=0\}}\big]x_j
\end{equation}
If we take a partial derivative for both sides of above equation with respect to the argument $x_i$, then since by Lemma \ref{l.01}, for sufficiently small $\varepsilon>0$, $\frac{\partial }{\partial x_i}\mathbb{P}[H(x)=0]=\frac{\partial }{\partial x_i}\mathbb{P}[|H(x)|\leq \varepsilon]=0$, we have that
\begin{equation}\label{2.14}
\mathbb{E}\bigg[\frac{\partial \varrho(x)}{\partial x_i}1_{\{H(x)=0\}}\bigg]=\sum_{j=1}^d\frac{\partial}{\partial x_i}\mathbb{E}\big[L_j1_{\{H(x)=0\}}\big]x_j+\mathbb{E}\big[L_i1_{\{H(x)=0\}}\big].
\end{equation}
From equation \eqref{2.14}, one can conclude that if the first term of the right--hand side of the equation equals zero, that is, $\sum_{j=1}^d\frac{\partial}{\partial x_i}\mathbb{E}\big[L_j1_{\{H(x)\}}\big]x_j=0$, then $\mathbb{E}\big[\frac{\partial H(x)}{\partial x_i}1_{\{H(x)=0\}}\big]=0$. In the following Proposition, we give results, which deal with the right--hand side of equation \eqref{2.14}.

\begin{prop}\label{p.01}
Let $(L_1,\dots,L_d)'$ be an integrable discrete random vector and $\mathbb{P}\big[H(x)=0\big]>0$. Then, for $k,i=1,\dots,d$, it holds
\begin{equation}\label{2.15}
\frac{\partial \varrho(x)}{\partial x_i}=\mathbb{E}\big[L_i\big|H(x)=0\big]
\end{equation}
and
\begin{equation}\label{2.16}
\frac{\partial^2 \varrho(x)}{\partial x_k\partial x_i}=0.
\end{equation}
\end{prop}

\begin{proof}
Since for $k\neq j$ and sufficiently small $\varepsilon>0$, by the iterated expectation formula,
\begin{equation}\label{2.17}
\frac{\partial}{\partial x_i}\mathbb{E}\big[L_j1_{\{H(x)=0\}}\big]=\frac{\partial}{\partial x_i}\mathbb{E}\big[L_j1_{\{|H(x)|\leq \varepsilon\}}\big]=\mathbb{E}\bigg[L_j\frac{\partial}{\partial x_i}\mathbb{P}\big[|H(x)|\leq \varepsilon\big|\tilde{L}_k\big]\bigg],
\end{equation}
according to equation \eqref{2.02} in Lemma \ref{l.01}, we find that
\begin{equation}\label{2.18}
\frac{\partial}{\partial x_i}\mathbb{E}\big[L_j1_{\{H(x)=0\}}\big]=0~~~\text{for}~i,j=1,\dots,d.
\end{equation}
Substituting equation \eqref{2.18} into equation \eqref{2.14}, we obtain the following equation
\begin{equation}\label{2.19}
\mathbb{E}\bigg[\frac{\partial \varrho(x)}{\partial x_i}1_{\{H(x)=0\}}\bigg]=\mathbb{E}\big[L_i1_{\{H(x)=0\}}\big].
\end{equation}
for $i\neq j$, $i,j=1,\dots,d$.
For the second order partial derivative, by Lemma \ref{l.02} and monotone convergence theorem, we have
\begin{equation}\label{2.20}
\mathbb{E}\bigg[\frac{\partial^2 \varrho(x)}{\partial x_k\partial x_i}1_{\{H(x)=0\}}\bigg]=\frac{\partial}{\partial x_k}\mathbb{E}\bigg[\frac{\partial \varrho(x)}{\partial x_i}1_{\{H(x)=0\}}\bigg]=\frac{\partial}{\partial x_k}\mathbb{E}\big[L_i1_{\{H(x)=0\}}\big]=0
\end{equation}
for $k,i=1,\dots,d$. 
Since $\mathbb{P}[H(x)=0]>0$, if we divide the both sides of equations \eqref{2.19} and \eqref{2.20} by $\mathbb{P}[H(x)=0]$, then we obtain equations \eqref{2.15} and \eqref{2.16}. That completes the proof of the Proposition.
\end{proof}

It should be noted that it follows from equation \eqref{2.16} that higher order partial derivatives of the risk measure $\varrho(x)$ are equal to zero. 

It seems that the first and second partial derivatives of VaR and ES for discrete portfolio loss random variable are first explored in \citeA{Rau-Bredow03}. However, it is worth mentioning that Proposition \ref{p.01} not only holds for the VaR and ES but also any risk measures. VaR and ES are two closely related and commonly used risk measures in practice. Mathematically, the risk measures at a confidence level $\alpha\in (0,1)$ are defined by 
\begin{equation}\label{2.21}
\text{VaR}_\alpha(x):=q_\alpha(x)=\inf\big\{t\in\mathbb{R}\big|F_{L(x)}(t)\geq \alpha\big\}
\end{equation}
and
\begin{equation}\label{2.22}
\text{ES}_\alpha(x):=\frac{1}{1-\alpha}\int_{\alpha}^1q_u(x)du,
\end{equation}
respectively, where $q_\alpha(x)$ is the $\alpha$--quantile and $F_{L(x)}(t)$ is a distribution function, respectively, of the portfolio loss random variable $L(x)$. Following the ideas in \citeA{Acerbi02} (see also \citeA{McNeil15}), for a generic portfolio loss random variable $L(x)$, it can be shown that
\begin{equation}\label{2.23}
\mathrm{ES}_\alpha(x)=\frac{1}{1-\alpha}\bigg(\mathbb{E}\big[L(x)1_{\{L(x)\geq q_\alpha(x)\}}\big]-q_\alpha(x)\Big(\mathbb{P}\big[L(x)\geq q_\alpha(x)\big]-(1-\alpha)\Big)\bigg).
\end{equation}
\citeA{Artzner99} introduced four axioms, namely
\begin{itemize}
\item[($i$)] Monotonicity. $L_1\leq L_2$ implies $\varrho(L_1)\leq \varrho(L_2)$.
\item[($ii$)] Translation invariance. For all $t\in \mathbb{R}$, $\varrho(L+t)=\varrho(L)+t$.
\item[($iii$)] Subadditivity. For all $L_1,L_2\in \mathcal{M}$, $\varrho(L_1+L_2)\leq \varrho(L_1)+\varrho(L_2)$. 
\item[($iv$)] Positive homogeneity. For all $\lambda\geq 0$, $\varrho(\lambda L)=\lambda\varrho(L)$.
\end{itemize}
\citeA{Artzner99} suppose that every risk measure should satisfy the four axioms and they refer to risk measure, satisfying the axioms as coherent. It can be shown that ES$_\alpha(x)$ is a coherent risk measure, see \citeA{Acerbi02} (see also \citeA{McNeil15}). By using the definition of the quantile, one can prove that VaR$_\alpha(x)$ satisfies the axioms: monotonicity, translation invariance, and positive homogeneity. But, in general, VaR$_\alpha(x)$ does not satisfy the third axiom subadditivity, see \citeA{McNeil15}. Therefore, it is not a coherent risk measure in general. However, it can be shown that for elliptically distributed loss random vector $(L_1,\dots,L_d)'$, the risk measure VaR$_\alpha(x)$ for $\alpha\in[0.5,1)$ satisfies the subadditivity axiom, see \citeA{McNeil05}.

Now, we consider some results, which deal with partial derivatives of the expected shortfall ES$_\alpha(x)$. If we take the generic risk measure $\varrho(x)$ by the Value--at--Risk, $\varrho(x)=\text{VaR}_\alpha(x)=q_\alpha(x)$, then the following Corollary holds.
\begin{cor}\label{c.01}
Let $(L_1,\dots,L_d)'$ be an integrable discrete random vector. Then, for $i,j=1,\dots,d$, it holds
\begin{equation}\label{2.24}
\frac{\partial}{\partial x_i}\mathrm{ES}_\alpha(x)=\frac{1}{1-\alpha}\bigg(\mathbb{E}\big[L_i1_{\{L(x)\geq q_\alpha(x)\}}\big]-\frac{\partial}{\partial x_i}q_\alpha(x)\Big(\mathbb{P}\big[L(x)\geq q_\alpha(x)\big]-(1-\alpha)\Big)\bigg)
\end{equation}
and
\begin{equation}\label{2.25}
\frac{\partial^2}{\partial x_i \partial x_j}\mathrm{ES}_\alpha(x)=0.
\end{equation}
\end{cor}

\begin{proof}
By the iterated expectation formula, it follows from equation \eqref{2.10} that
\begin{equation}\label{2.26}
\frac{\partial}{\partial x_i}\mathbb{P}\big[L(x)\geq q_\alpha(x)\big]=0.
\end{equation}
Therefore, due to equations \eqref{2.05}, \eqref{2.23}, and \eqref{2.26}, one obtains equation \eqref{2.24}. For the second order derivatives of the expected shortfall ES$_\alpha(x)$, by taking account equation \eqref{2.04}, equation \eqref{2.16}, and \eqref{2.26} for equation \eqref{2.24}, one gets equation \eqref{2.25}.
\end{proof}

\section{Continuous Portfolio Loss Random Variable}

In this Section, we consider partial derivatives of risk measures of absolute continuous portfolio loss random variable $L(x)$. We denote for given $\tilde{L}_j$, a conditional density function of the random variable $L_j$ by $f_{L_j|\tilde{L}_j}(t)$ for $j=1,\dots,d$, a density function and cumulative distribution function of the random variable $H(x)$ by $f_{H(x)}(t)$ and $F_{H(x)}(t)$, respectively, and the right tail probability of $H(x)$ by $\bar{F}_{H(x)}(t):=\mathbb{P}[H(x)>t]=1-F_{H(x)}(t)$. We assume that second order partial derivatives of the tail probability exist and are integrable. That is, the following Assumption holds
\begin{ass}\label{a.03}
The second order partial derivatives of the tail probability $\bar{F}_{H(x)}(t)$ with respect to the argument $x$ exist and the partial derivatives are integrable with respect to the argument $t$.
\end{ass}
Now, we give a main Lemma, which plays a major role in calculating partial derivatives of risk measures of absolute continuous portfolio loss random variable $L(x)$. 

\begin{lemma}\label{l.02}
Let for each $j=1,\dots,d$, $g\big(\tilde{L}_j\big)$ be an integrable random variable, where $g:\mathbb{R}^{d-1}\to\mathbb{R}$ is a Borel function. Then, for $j=1,\dots,d$ and $t\in\mathbb{R}$, we have
\begin{equation}\label{3.01}
\frac{1}{|x_j|} \mathbb{E}\bigg[g\big(\tilde{L}_j\big)f_{L_j|\tilde{L}_j}\bigg(\frac{\varrho(x)+t-\tilde{L}_j(x)}{x_j}\bigg)\bigg]=f_{H(x)}(t)\mathbb{E}\big[g\big(\tilde{L}_j\big)\big|H(x)=t\big].
\end{equation}
\end{lemma}

\begin{proof}
See \citeA{McNeil15}.
\end{proof}

It follows from Lemma \ref{l.02} that the following Proposition holds.
\begin{prop}\label{p.02}
Let $(L_1,\dots,L_d)'$ be an integrable jointly continuous random vector. Then, for $i=1,\dots,d$ and $t\in\mathbb{R}$, it holds
\begin{equation}\label{3.02}
\frac{\partial \varrho(x)}{\partial x_i}=\mathbb{E}\big[L_i\big|H(x)=t\big]-\frac{1}{f_{H(x)}(t)}\frac{\partial \bar{F}_{H(x)}(t)}{\partial x_i},
\end{equation}
\begin{equation}\label{3.03}
\frac{\partial \varrho(x)}{\partial x_i}=\mathbb{E}\big[L_i\big|H(x)\geq t(x)\big]-\frac{1}{\bar{F}_{H(x)}(t(x))}\int_{t(x)}^\infty\frac{\partial \bar{F}_{H(x)}(z)}{\partial x_i}dz,
\end{equation}
and
\begin{equation}\label{3.04}
\frac{\partial \varrho(x)}{\partial x_i}=\mathbb{E}\big[L_i\big]-\int_{-\infty}^\infty\frac{\partial \bar{F}_{H(x)}(z)}{\partial x_i}dz.
\end{equation}
\end{prop}

\begin{proof}
Without loss of generality, let us suppose $x_j>0$. By the iterated expectation formula and Lemma \ref{l.02}, one gets that for $j\neq i$,
\begin{eqnarray}\label{3.05}
\frac{\partial F_{H(x)}(t)}{\partial x_i}&=&\frac{\partial}{\partial x_i}\mathbb{E}\bigg[\mathbb{P}\bigg(L_j\leq \frac{\varrho(x)+t-\tilde{L}_j(x)}{x_j}\bigg|\tilde{L}_j\bigg)\bigg]\nonumber\\
&=&\frac{\partial \varrho(x)}{\partial x_i}\frac{1}{x_j}\mathbb{E}\bigg[f_{L_j|\tilde{L}_j}\bigg(\frac{\varrho(x)+t-\tilde{L}_j(x)}{x_j}\bigg)\bigg]-\frac{1}{x_j}\mathbb{E}\bigg[L_if_{L_j|\tilde{L}_j}\bigg(\frac{\varrho(x)+t-\tilde{L}_j(x)}{x_j}\bigg)\bigg]\nonumber\\
&=&\frac{\partial \varrho(x)}{\partial x_i}f_{H(x)}(t)-f_{H(x)}(t)\mathbb{E}\big[L_i\big|H(x)=t\big].
\end{eqnarray}
Thus equation \eqref{3.02} holds. If we integrate the above equation from $t(x)$ to positive infinity and from negative infinity to positive infinity, then we obtain equations \eqref{3.03} and \eqref{3.04}, respectively.
\end{proof}

For the above Proposition, note that $t(x)$ in equation \eqref{3.03} can depend on the argument $x$, while $t$ in equation \eqref{3.02} can not depend on the argument $x$. Henceforth, we will use the notations $t(x)$ and $t$ for the same reason. Since for the risk measure Value--at--Risk $\varrho(x)=q_\alpha(x)$, $F_{H(x)}(0)=\alpha$, we get the well--known formula of the first order derivative of the Value--at--Risk, namely, $\frac{\partial }{\partial x_i}q_\alpha(x)=\mathbb{E}\big[L_i\big|L(x)=q_\alpha(x)\big]$, see \citeA{Tasche99,Tasche00}, \citeA{Gourieroux00}, and \citeA{McNeil15}. As we compare equations \eqref{2.15} and \eqref{3.02}, for absolute continuous portfolio loss random variable, there is an additional term arose. If we ignore the additional terms, equations \eqref{3.02} and \eqref{3.03} look like partial derivatives of VaR$_\alpha(x)$ and ES$_\alpha(x)$, respectively, for the partial derivatives of ES$_\alpha(x)$, see below. Applying the idea in proof of Proposition \ref{p.02}, one can develop formulas of the first order partial derivatives of any risk measure for a nonlinear portfolio loss random variable. It is worth mentioning that equations \eqref{3.02} and \eqref{3.03} are new formulas for the calculation of the first order derivatives of any risk measures. 

Now we consider a case, where the risk measure $\varrho(x)$ is a positive homogeneous function. If we multiply equation \eqref{3.02} by $x_i$ and then sum it for all $i=1,\dots,d$, then we find that
\begin{eqnarray}\label{3.06}
\varrho(x)&=&\sum_{i=1}^d\frac{\partial \varrho(x)}{\partial x_i}x_i=\mathbb{E}\big[L(x)\big|H(x)=t\big]-\frac{1}{f_{H(x)}(t)}\sum_{i=1}^d\frac{\partial \bar{F}_{H(x)}(t)}{\partial x_i}x_i\nonumber\\
&=&\varrho(x)+t-\frac{1}{f_{H(x)}(t)}\sum_{i=1}^d\frac{\partial \bar{F}_{H(x)}(t)}{\partial x_i}x_i.
\end{eqnarray}
As a result, we obtain that
\begin{equation}\label{3.07}
\sum_{i=1}^d\frac{\partial \bar{F}_{H(x)}(t)}{\partial x_i}x_i=tf_{H(x)}(t).
\end{equation}
In special case, corresponding to $t=0$ for the above equation, $F_{H(x)}(0)$ becomes a homogeneous function with degrees of zero. Using the idea in proof of Proposition \ref{p.02}, for a Borel function $h:\mathbb{R}\to\mathbb{R}$, which has derivative and the positive homogeneous risk measure $\varrho(x)$, it can be shown that $h(\varrho(x))$ has the following representation
\begin{equation}\label{3.08}
h(\varrho(x))=h'(\varrho(x))\varrho(x)+\frac{1}{f_{L(x)}(h(\varrho(x)))}\sum_{i=1}^d\frac{\partial \bar{F}_{L(x)}(h(\varrho(x)))}{\partial x_i}x_i.
\end{equation}
If we choose the function $h(t)$ by the logarithm and exponential functions, which are strictly monotone functions, i.e., $h(t)=\ln(t)$ and $h(t)=\exp(t)$, then we have that
\begin{equation}\label{3.09}
\ln(\varrho(x))=1+\frac{1}{f_{L(x)}(\ln(\varrho(x)))}\sum_{i=1}^d\frac{\partial \bar{F}_{L(x)}(\ln(\varrho(x)))}{\partial x_i}x_i
\end{equation}
and
\begin{equation}\label{3.10}
\varrho(x)=1-\frac{1}{\exp(\varrho(x))f_{L(x)}(\exp(\varrho(x)))}\sum_{i=1}^d\frac{\partial \bar{F}_{L(x)}(\exp(\varrho(x)))}{\partial x_i}x_i.
\end{equation}

In the following Proposition, we consider results, which depend on the second order moments of the loss random vector $(L_1,\dots,L_d)'$.

\begin{prop}\label{p.03}
Let $(L_1,\dots,L_d)'$ be an square integrable continuous random vector and for $i,j=1,\dots,d$, $\frac{\partial}{\partial x_i}\mathbb{E}\big[L_j1_{\{H(x)\geq t\}}\big]$ be integrable functions with respect to the argument $t$. Then, for $t\in \mathbb{R}$ and $i,j=1,\dots,d$, it holds
\begin{equation}\label{3.11}
\frac{\partial \varrho(x)}{\partial x_i}=\frac{1}{\mathbb{E}\big[L_j1_{\{H(x)\geq t(x)\}}\big]}\bigg(\mathbb{E}\big[L_iL_j1_{\{H(x)\geq t(x)\}}\big]-\int_{t(x)}^\infty\frac{\partial }{\partial x_i}\mathbb{E}\big[L_j1_{\{H(x)\geq z\}}\big]dz\bigg),
\end{equation}
\begin{equation}\label{3.12}
\frac{\partial}{\partial x_i}\mathbb{E}\big[L_j1_{\{H(x)\geq t\}}\big]=f_{H(x)}(t)\mathrm{Cov}\big[L_i,L_j\big|H(x)=t\big]+\mathbb{E}\big[L_j\big|H(x)=t\big]\frac{\partial \bar{F}_{H(x)}(t)}{\partial x_i},
\end{equation}
\begin{equation}\label{3.13}
\frac{\partial}{\partial x_i}\mathbb{E}\big[L(x)1_{\{H(x)\geq t\}}\big]=\mathbb{E}\big[L_i1_{\{H(x)\geq t\}}\big]+\big(\varrho(x)+t\big)\frac{\partial \bar{F}_{H(x)}(t)}{\partial x_i},
\end{equation}
and
\begin{eqnarray}\label{3.14}
\frac{\partial \varrho(x)}{\partial x_i\partial x_j}&=&\frac{1}{\bar{F}_{H(x)}(t)}\bigg[f_{H(x)}(t)\mathrm{Cov}\big[L_i,L_j\big|H(x)=t\big]\nonumber\\
&+&\frac{1}{f_{H(x)}(t)}\frac{\partial \bar{F}_{H(x)}(t)}{\partial x_i}\frac{\partial \bar{F}_{H(x)}(t)}{\partial x_j}-\int_t^\infty\frac{\partial^2 \bar{F}_{H(x)}(z)}{\partial x_i\partial x_j}dz\bigg].
\end{eqnarray}
\end{prop}

\begin{proof}
Without loss of generality, we suppose $x_k>0$. According to the iterated expectation formula and Lemma \ref{l.02}, for $k\neq i,j$, the partial derivative $\frac{\partial}{\partial x_i}\mathbb{E}\big[L_j1_{\{H(x)\geq t\}}\big]$ is obtained by
\begin{eqnarray}\label{3.15}
\frac{\partial}{\partial x_i}\mathbb{E}\big[L_j1_{\{H(x)\geq t\}}\big]&=&-\frac{\partial}{\partial x_i}\mathbb{E}\bigg[L_j \mathbb{P}\bigg(L_k\leq \frac{\varrho(x)+t-\tilde{L}_k(x)}{x_k}\bigg|\tilde{L}_k\bigg)\bigg]\nonumber\\
&=&f_{H(x)}(t)\mathbb{E}\big[L_iL_j\big|H(x)=t\big]-\frac{\partial\varrho(x)}{\partial x_i}f_{H(x)}(t) \mathbb{E}\big[L_j\big|H(x)=t\big].
\end{eqnarray}
Integrating equation \eqref{3.15} from $t(x)$ to positive infinity, we get equation \eqref{3.11}. Taking into account equation \eqref{3.02}, we obtain equation \eqref{3.12}. Consequently, it follows from equation \eqref{3.12} that the sum in equation \eqref{2.12}, which is adjusted by the argument $t$ is given by
\begin{eqnarray}\label{3.16}
\sum_{j=1}^d\frac{\partial}{\partial x_i}\mathbb{E}\big[L_j1_{\{H(x)\geq t\}}\big]x_j&=&f_{H(x)}(t)\mathrm{Cov}\big[L_i,L(x)\big|H(x)=t\big]\nonumber\\
&+&\mathbb{E}\big[L(x)\big|H(x)=t\big]\frac{\partial \bar{F}_{H(x)}(t)}{\partial x_i}\nonumber\\
&=&\big(\varrho(x)+t\big)\frac{\partial \bar{F}_{H(x)}(t)}{\partial x_i}.
\end{eqnarray}
As a result, due to equation \eqref{2.12}, which is adjusted by the argument $t$, one obtains equation \eqref{3.13}. By equations \eqref{3.02}, \eqref{3.03}, and \eqref{3.12}, we find that
\begin{eqnarray}\label{3.17}
\frac{\partial}{\partial x_i}\mathbb{E}\big[L_j\big|H(x)\geq t\big]&=&\frac{\partial}{\partial x_i}\bigg\{\frac{1}{\bar{F}_{H(x)}(t)}\mathbb{E}\big[L_j1_{\{H(x)\geq t\}}\big]\bigg\}\nonumber\\
&=&\frac{1}{\bar{F}_{H(x)}(t)}\bigg[f_{H(x)}(t)\mathrm{Cov}\big[L_i,L_j\big|H(x)=t\big]\\
&+&\bigg(\frac{1}{f_{H(x)}(t)}\frac{\partial \bar{F}_{H(x)}(t)}{\partial x_j}-\frac{1}{\bar{F}_{H(x)}(t)}\int_t^\infty\frac{\partial \bar{F}_{H(x)}(z)}{\partial x_j}dz\bigg)\frac{\partial \bar{F}_{H(x)}(t)}{\partial x_i}\bigg]\nonumber.
\end{eqnarray} 
Thus, from equations \eqref{3.03} and \eqref{3.17} and a fact that
\begin{eqnarray}\label{3.18}
\frac{\partial}{\partial x_i}\bigg\{\frac{1}{\bar{F}_{H(x)}(t)}\int_t^\infty\frac{\partial \bar{F}_{H(x)}(z)}{\partial x_j}dz\bigg\}&=&\frac{1}{\bar{F}_{H(x)}(t)}\int_t^\infty\frac{\partial^2 \bar{F}_{H(x)}(z)}{\partial x_i\partial x_j}dz\\
&-&\frac{1}{\bar{F}_{H(x)}^2(t)}\int_t^\infty\frac{\partial \bar{F}_{H(x)}(z)}{\partial x_j}dz\frac{\partial \bar{F}_{H(x)}(t)}{\partial x_i}\nonumber
\end{eqnarray}
we obtain equation \eqref{3.14}. 
\end{proof}

It should be noted that for elliptically distributed loss random vector $(L_1,\dots,L_d)'$, the first and second order derivatives of any risk measure can be easily calculated from the representation of the risk measure, see \citeA{McNeil15}.

Let us consider the following mean--risk measure portfolio choice problem
\begin{equation}\label{3.19}
\begin{cases}
\varrho(x)\longrightarrow \min\\
\text{s.t.}~x'\mu=r_p~\text{and}~x'e=1,
\end{cases}
\end{equation}
where $\mu=(\mathbb{E}[L_1],\dots,\mathbb{E}[L_d])'$ is an expectation vector of the loss random vector, $e=(1,\dots,1)'\in\mathbb{R}^d$ is a vector, consisting of 1, and $r_p$ is a portfolio return, which is defined by investors. Note that if we replace the objective function $\varrho(x)$ with $x'\Sigma x$, then we obtain \citeauthor{Markowitz52}'s \citeyear{Markowitz52} mean--variance portfolio choice problem, where $\Sigma$ is a covariance matrix of the loss random vector. For the mean--risk measure portfolio choice problem, the convexity of the risk measure (objective function) $\varrho(x)$ is a crucial topic. Since the constraint set is a convex set, if the risk measure is a convex function, then the problem becomes a convex optimization problem. Convex optimization problems have nice properties. Thus, one should check that the risk measure is a convex function. One method to check the convexity of the risk measure is that the Hessian matrix of the risk measure is positive semi--definite. For this reason, let us denote the Hessian matrix of the risk measure $\varrho(x)$ by $R(x)$, that is, $R(x):=\big(\frac{\partial^2}{\partial x_i\partial x_j}\varrho(x)\big)_{i,j=1}^d$. Then, the following Corollary holds, which checks the convexity of the risk measure:
\begin{cor}\label{c.02}
Let $(L_1,\dots,L_d)'$ be a square integrable continuous random vector. If at least one of the following conditions hold, then the Hessian $R(x)$ of a risk measure $\varrho(x)$ is positive semi--definite:
\begin{itemize}
\item[(i)] for all $t\in \mathbb{R}$, the Hessian of the tail probability $\bar{F}_{H(x)}(t)$, $\big(\frac{\partial^2}{\partial x_i\partial x_j}\bar{F}_{H(x)}(t)\big)_{i,j=1}^d$ is negative semi--definite,
\item[(ii)] there is a $t\in \mathbb{R}$ such that a matrix
\begin{eqnarray}\label{3.20}
&&\bigg(\frac{1}{\bar{F}_{H(x)}(t)}\bigg[f_{H(x)}(t)\mathrm{Cov}\big[L_i,L_j\big|H(x)=t\big]\nonumber\\
&&+\frac{1}{f_{H(x)}(t)}\frac{\partial \bar{F}_{H(x)}(t)}{\partial x_i}\frac{\partial \bar{F}_{H(x)}(t)}{\partial x_j}-\frac{\partial^2 }{\partial x_i\partial x_j}\mathbb{E}\big[(H(x)-t)^+\big]\bigg]\bigg)_{i,j=1}^d
\end{eqnarray}
is positive semi--definite, where for a real number $x\in\mathbb{R}$, $x^+:=\max(x,0)$ is a maximum of $x$ and 0.
\end{itemize}
\end{cor}

\begin{proof}
($i$) follows from equation \eqref{3.04}. For $(ii)$, since $L(x)$ is an integrable continuous random variable, by the integration by parts, we have that
\begin{equation}\label{3.21}
\int_t^\infty \bar{F}_{H(x)}(z)dz=\mathbb{E}\big[(H(x)-t)^+\big].
\end{equation}
Consequently, equation \eqref{3.14} implies equation \eqref{3.20}.
\end{proof}

Note that $\big(\mathrm{Cov}\big[L_i,L_j\big|H(x)=t\big]\big)_{i,j=1}^d$ is a positive semi--definite matrix. Therefore, to check the Hessian $R(x)$ is positive semi--definite, one may check the second line matrices of equation \eqref{3.20} is positive semi--definite for some $t\in\mathbb{R}$.

Finally, in the following Corollary, which is a direct consequence of Proposition \ref{p.03}, we give results that deal with partial derivatives of the expected shortfall ES$_\alpha(x)$.

\begin{cor}\label{c.07}
Let $(L_1,\dots,L_d)'$ be a square integrable continuous random vector. Then, for $i,j=1,\dots,d$, it holds
\begin{equation}\label{3.22}
\frac{\partial}{\partial x_i}\mathrm{ES}_\alpha(x)=\frac{1}{1-\alpha}\mathbb{E}\big[L_i1_{\{L(x)\geq q_\alpha(x)\}}\big]=\mathbb{E}\big[L_i\big|L(x)\geq q_\alpha(x)\big]
\end{equation}
and
\begin{equation}\label{3.23}
\frac{\partial^2}{\partial x_i \partial x_j}\mathrm{ES}_\alpha(x)=\frac{1}{1-\alpha}f_{H(x)}(0)\mathrm{Cov}\big[L_i,L_j\big|L(x)=q_\alpha(x)\big].
\end{equation}
\end{cor}

\begin{proof}
According to equations \eqref{2.23} and \eqref{3.13} and the fact that $\frac{\partial}{\partial x_i}\bar{F}_{H(x)}(0)=\frac{\partial}{\partial x_i}\bar{F}_{L(x)}(q_\alpha(x))=0$, we get equation \eqref{3.22}. It follows from equations \eqref{3.12} and \eqref{3.22} that equation \eqref{3.23} holds.
\end{proof}

It should be emphasized that since $\big(\mathrm{Cov}\big[L_i,L_j\big|L(x)=q_\alpha(x)\big]\big)_{i,j=1}^d$ is a positive semi--definite matrix, the risk measure ES$_\alpha(x)$ is a convex function with respect to the argument $x$. If we take $t(x)=q_\alpha(x)-\varrho(x)$ in equation \eqref{3.03}, then equation \eqref{3.22} implies that
\begin{equation}\label{3.24}
\frac{\partial \varrho(x)}{\partial x_i}=\frac{\partial}{\partial x_i}\mathrm{ES}_\alpha(x)-\frac{1}{1-\alpha}\int_{q_\alpha(x)-\varrho(x)}^\infty\frac{\partial \bar{F}_{L(x)}(\varrho(x)+z)}{\partial x_i}dz.
\end{equation}
In particular, if $\varrho(x)=\text{ES}_\alpha(x)$, then we have that
\begin{equation}\label{3.25}
\int_{q_\alpha(x)-\text{ES}_\alpha(x)}^\infty\frac{\partial \bar{F}_{L(x)}(\text{ES}_\alpha(x)+z)}{\partial x_i}dz=0.
\end{equation}

\section{Heavy Tailed Distributions}

Extreme value theory studies the stochastic behavior of the extreme values in a process. In the univariate case, the stochastic behavior of the maxima of independent identically distributed (i.i.d.) random variables can be described by the three extreme value distributions, namely, Fr$\acute{\mathrm{e}}$chet, Weibull, and Gumbel. There are two main methods that model the extreme values. The first method relies on block maxima of the i.i.d. random variables, while the second method relies on peak values above a certain threshold. Because the second method applies more data on extreme outcomes than the first method, the second method is considered most useful in practice.

Let $L_1,L_2,\dots$ be an i.i.d. sequence of the loss random variables. Then, the block maxima of the sequence is defined by
\begin{equation}\label{4.01}
M_n:=\max(L_1,\dots,L_n).
\end{equation}
Because a distribution function of the block maxima converges to a degenerate distribution, we need to seek normalizing constants $c_n>0$ and $d_n\in\mathbb{R}$ such that a limiting distribution of the block maxima is non--degenerate $H(t)$, i.e.,
\begin{equation}\label{4.02}
\lim_{n\to\infty}\mathbb{P}\big[c_n^{-1}(M_n-d_n)\leq t\big]=H(t).
\end{equation}
According the Fisher--Tippett theorem (see \citeA{Embrechts97} and \citeA{McNeil15}), if there exist normalizing constants $c_n>0$ and $d_n\in\mathbb{R}$, and non--degenerate distribution function $H(t)$, then $H(t)$ must be type of one of the following distributions:
\begin{equation}\label{4.03}
\begin{matrix}
\text{Fr$\acute{\mathrm{e}}$chet:} & & \Phi_{\kappa}(t) &=& \begin{cases}
0, & t\leq 0\\
\exp\{-t^{-\kappa}\}, & t>0
\end{cases} & \kappa>0.\\
\text{Weibull:} & & \Psi_{\kappa}(t) &=& \begin{cases}
\exp\{-(-t)^{\kappa}\}, & t\leq 0\\
1, & t>0
\end{cases} & \kappa>0.\\
\text{Gumbel:} & & \Lambda(t) &=& \exp\{-\exp\{-t\}\}, & t\in \mathbb{R}. & \\
\end{matrix}
\end{equation}
The rest of the paper, we consider the Fr$\acute{\mathrm{e}}$chet case, which is commonly used to model tails of financial loss random variables.

\subsection{$t$ Does not Depend on $x$}

In this Subsection, we assume that $t$ does not depend on the variable $x$. Due to the extreme value theory, it is the well--known fact that for a heavy--tailed (Fr$\acute{\mathrm{e}}$chet case) random variable $L(x)$, its tail probability $\bar{F}_{L(x)}(t)$ is represented by
\begin{equation}\label{4.04}
\bar{F}_{L(x)}(t)=t^{-\kappa}\ell_{L(x)}(t),~~~t\to\infty
\end{equation}
for $\kappa>0$, where $\ell_{L(x)}(t)$ is a slowly varying function at $\infty$, that is, it satisfy the following condition 
\begin{equation}\label{4.05}
\lim_{t\to\infty}\frac{\ell_{L(x)}(\lambda t)}{\ell_{L(x)}(t)}=1~~~\text{for all}~\lambda>0,
\end{equation}
see \citeA{Embrechts97} and \citeA{McNeil15}. The parameter $\kappa$ is often referred to as the tail index of the distribution function $F_{L(x)}(t)$. From representation \eqref{4.04} of the tail probability, one may deduce that the tail probability decays like a power function. In this paper, we assume that the density function $f_{L(x)}(t)$ is an ultimately monotone function, that is, there is an interval $(a,\infty)$ for some $a>0$ such that $f_{L(x)}(t)$ is monotone. Under the assumption, by monotone density theorem, the density function $f_{L(x)}(t)$ is represented by 
\begin{equation}\label{4.06}
f_{L(x)}(t)=\kappa t^{-k-1}\ell_{L(x)}(t),~~~t\to\infty,
\end{equation}
see \citeA{Bingham89} and \citeA{Embrechts97}. In the following lemma, we show that the conditional expectation $\mathbb{E}\big[L_i\big|L(x)=t\big]$ is a regular varying function at $\infty$ with index 1.

\begin{prop}\label{p.04}
Let tail probability of a random variable $L(x)$ is represented by $\bar{F}_{L(x)}(t)=t^{-\kappa}\ell_{L(x)}(t)$ for $\kappa>1$, where $\ell_{L(x)}(t)$ is slowly varying function at $\infty$. Then, $\mathbb{E}\big[L_i\big|L(x)=t\big]$ is a regular varying function at $\infty$ with index 1
\end{prop}

\begin{proof}
If we take $\varrho(x)=0$ in equation \eqref{3.02}, then we have that
\begin{equation}\label{4.07}
\frac{\mathbb{E}\big[L_i\big|L(x)=\lambda t\big]}{\mathbb{E}\big[L_i\big|L(x)=t\big]}=\frac{f_{L(x)}(t)}{f_{L(x)}(\lambda t)}\frac{\frac{\partial }{\partial x}\bar{F}_{L(x)}(\lambda t)}{\frac{\partial }{\partial x}\bar{F}_{L(x)}( t)}.
\end{equation}
By substituting equations \eqref{4.04} and \eqref{4.06} into the above equation, one gets that
\begin{equation}\label{4.07}
\frac{\mathbb{E}\big[L_i\big|L(x)=\lambda t\big]}{\mathbb{E}\big[L_i\big|L(x)=t\big]}=\lambda\frac{\ell_{L(x)}(t)}{\ell_{L(x)}(\lambda t)}\frac{\frac{\partial }{\partial x}\ell_{L(x)}(\lambda t)}{\frac{\partial }{\partial x}\ell_{L(x)}(t)},~~~t\to\infty.
\end{equation}
According to equation \eqref{4.05}, one obtains
\begin{equation}\label{4.10}
\lim_{t\to\infty}\frac{\mathbb{E}\big[L_i\big|L(x)=\lambda t\big]}{\mathbb{E}\big[L_i\big|L(x)=t\big]}=\lambda.
\end{equation}
Thus, the conditional expectation $\mathbb{E}\big[L_i\big|L(x)=t\big]$ is a regularly varying function at $\infty$ with index 1. That completes the proof.
\end{proof}

According to Proposition \ref{p.04}, the conditional expectation $\mathbb{E}\big[L_i\big|L(x)=t\big]$ is represented by
\begin{equation}\label{4.11}
\mathbb{E}\big[L_i\big|L(x)=t\big]=t\ell_{L(x)}^i(t),~~~t\to\infty
\end{equation}
for $i=1,\dots,d$, where $\ell_{L(x)}^i(t)$ is a slowly varying function at $\infty$.

\begin{prop}\label{p.05}
Let tail probability of a random variable $L(x)$ is represented by $\bar{F}_{L(x)}(t)=t^{-\kappa}\ell_{L(x)}(t)$ for $\kappa>1$, for $i=1,\dots,d$, conditional expectations $\mathbb{E}\big[L_i\big|L(x)=t\big]$ are represented by $\mathbb{E}\big[L_i\big|L(x)=t\big]=t\ell_{L(x)}^i(t)$, and density $f_{L(x)}(t)$ be an ultimately monotone function, where $\ell_{L(x)}(t)$ and $\ell_{L(x)}^i(t)$ for $i=1,\dots,d$ are slowly varying function at $\infty$. Then, for the slowly varying functions, the following results hold
\begin{equation}\label{4.12}
\ell_{L(x)}^i(t)=\frac{1}{x_1+\dots+x_d},~~~t\to\infty
\end{equation}
for $i=1,\dots,d$ and
\begin{equation}\label{4.13}
\ell_{L(x)}(t)=(x_1+\dots+x_d)^\kappa,~~~t\to\infty.
\end{equation}
Consequently, we have that
\begin{equation}\label{4.14}
\bar{F}_{L(x)}(t)=t^{-\kappa}(x_1+\dots+x_d)^\kappa,~~~t\to\infty
\end{equation}
and
\begin{equation}\label{4.15}
f_{L(x)}(t)=\kappa t^{-\kappa-1}(x_1+\dots+x_d)^\kappa,~~~t\to\infty.
\end{equation}
\end{prop}

\begin{proof}
Since $\bar{F}_{L(x)}(t)=t^{-\kappa}\ell_{L(x)}(t)$ for $\kappa>1$ and the density $f_{L(x)}(t)$ is the ultimate monotone function, by monotone density theorem, the density has representation $f_{L(x)}(t)=\kappa t^{-(\kappa+1)}\ell_{L(x)}(t)$. If we multiply equation $\mathbb{E}\big[L_i\big|L(x)=t\big]=t\ell_{L(x)}^i(t)$ by $x_i$, sum it for $i=1,\dots,d$, and use a fact that $\mathbb{E}\big[L(x)\big|L(x)=t\big]=t$, then one obtains that
\begin{equation}\label{4.16}
\sum_{i=1}^d\ell_{L(x)}^i(t)x_i=1.
\end{equation}
According to equation \eqref{3.02}, we have
\begin{equation}\label{4.17}
\frac{\partial \bar{F}_{L(x)}(t)}{\partial x_i}=\mathbb{E}\big[L_i\big|L(x)=t\big]f_{L(x)}(t)
\end{equation}
for $i=1,\dots,d$. Therefore, it follows from representations of the tail probability, density function, and conditional expectation of the portfolio loss random variable $L(x)$ that
\begin{equation}\label{4.18}
t^{-\kappa}\frac{\partial \ell_{L(x)}(t)}{\partial x_i}=\Big(t\ell_{L(x)}^i(t)\Big)\times \Big(\kappa t^{-(\kappa+1)}\ell_{L(x)}(t)\Big)
\end{equation} 
for $i=1,\dots,d$. As a result, we find that
\begin{equation}\label{4.19}
\frac{\partial \ell_{L(x)}(t)}{\partial x_i}=\kappa\ell_{L(x)}^i(t)\ell_{L(x)}(t)
\end{equation} 
for $i=1,\dots,d$. Consequently, due to equation \eqref{4.16}, we reach the following problem for $\ell_{L(x)}(t)$ and $\ell_{L(x)}^1(t),\dots,\ell_{L(x)}^d(t)$:
\begin{equation}\label{4.20}
\begin{cases}
\dfrac{\partial \ln\big(\ell_{L(x)}(t)\big)}{\partial x_i}=\kappa\ell_{L(x)}^i(t)~~~\text{for}~i=1,\dots,d,\\
\ell_{L(x)}^1(t)x_1+\dots+\ell_{L(x)}^d(t)x_d=1.
\end{cases}
\end{equation}
One can easily check that a solution of the problem is given by the following equations
\begin{equation}\label{4.21}
\ell_{L(x)}^i(t)=\frac{1}{x_1+\dots+x_d}
\end{equation}
for $i=1,\dots,d$ and
\begin{equation}\label{4.22}
\ell_{L(x)}(t)=(x_1+\dots+x_d)^\kappa.
\end{equation}
Equations \eqref{4.14} and \eqref{4.15} follows from representations of the tail probability $\bar{F}_{L(x)}(t)$ and the density function $f_{L(x)}(t)$ and equation \eqref{4.13}.
\end{proof}

\begin{cor}\label{c.04}
Let conditions of Proposition \ref{p.05} hold. Then, the following results hold
\begin{equation}\label{4.23}
\mathbb{E}\big[L_i\big|L(x)=t\big]=\frac{t}{x_1+\dots+x_d},~~~t\to\infty
\end{equation}
and
\begin{equation}\label{4.24}
\mathbb{E}\big[L_i\big|L(x)\geq t\big]=\frac{\kappa}{\kappa-1}\frac{t}{x_1+\dots+x_d},~~~t\to\infty
\end{equation}
for $i=1,\dots,d$.
\end{cor}

\begin{proof}
Equations \eqref{4.11} and \eqref{4.13} imply equation \eqref{4.23}. It follows from equation \eqref{4.14} that the partial derivative of the tail probability with respect to the argument $x_i$ is given by
\begin{equation}\label{4.25}
\frac{\partial \bar{F}_{L(x)}(t)}{\partial x_i}=t^{-\kappa}\kappa(x_1+\dots+x_d)^{\kappa-1}.
\end{equation}
Thus, one obtains that for $\kappa>1$,
\begin{equation}\label{4.26}
\int_{t(x)}^\infty \frac{\partial \bar{F}_{L(x)}(z)}{\partial x_i}dz= {t(x)}^{-(\kappa-1)}\frac{\kappa}{\kappa-1}(x_1+\dots+x_d)^{\kappa-1},~~~t\to\infty.
\end{equation}
Consequently, if take $\varrho(x)=0$ in equation \eqref{4.03}, then by equations \eqref{4.14}/\eqref{4.44} (with $\varrho(x)=t(x)$, see below) and \eqref{4.26}, we find that
\begin{equation}\label{4.27}
\mathbb{E}\big[L_i\big|L(x)\geq t(x)\big]=\frac{1}{\bar{F}_{L(x)}(t(x))}\int_{t(x)}^\infty\frac{\partial \bar{F}_{L(x)}(z)}{\partial x_i}dz=\frac{\kappa}{\kappa-1}\frac{t(x)}{x_1+\dots+x_d},~~~t\to\infty
\end{equation}
for $i=1,\dots,d$. Thus, by taking $t(x)=t$ in the above equation, one obtains equation \eqref{4.24}.
\end{proof}

Now we give some asymptotic results, which are related to the second order conditional moments of components of the jointly continuous random vector $(L_1,\dots,L_d)'$.

\begin{cor}\label{c.05}
Let $\kappa>2$ and conditions of Proposition \ref{p.05} hold. Then, the following results hold
\begin{equation}\label{4.28}
\mathbb{E}\big[L_iL_j\big|L(x)=t\big]= \frac{t^2}{(x_1+\dots+x_d)^2},~~~t\to\infty
\end{equation}
and
\begin{equation}\label{4.29}
\mathbb{E}\big[L_iL_j\big|L(x)\geq t\big]=  \frac{\kappa}{\kappa-2}\frac{t^2}{(x_1+\dots+x_d)^2},~~~t\to\infty
\end{equation}
for $i,j=1,\dots,d$.
\end{cor}

\begin{proof}
Observe that
\begin{equation}\label{4.30}
\frac{\partial }{\partial x_i}\mathbb{E}\big[L_j1_{\{L(x)\geq t\}}\big]=\frac{\partial \bar{F}_{L(x)}(t)}{\partial x_i}\mathbb{E}\big[L_j\big|L(x)\geq t\big]+\bar{F}_{L(x)}(t)\frac{\partial }{\partial x_i}\mathbb{E}\big[L_j\big|L(x)\geq t\big]
\end{equation}
If we substitute equations \eqref{4.14}, \eqref{4.24}, and \eqref{4.25} into the above equation, we get that
\begin{equation}\label{4.31}
\frac{\partial }{\partial x_i}\mathbb{E}\big[L_j1_{\{L(x)\geq t\}}\big]= t^{-k+1}\kappa(x_1+\dots+x_d)^{k-2},~~~t\to\infty.
\end{equation}
By equations \eqref{4.23} and \eqref{4.25}, one finds that
\begin{equation}\label{4.32}
\frac{\partial \bar{F}_{L(x)}(t)}{\partial x_i}\mathbb{E}\big[L_j\big|L(x)= t\big]= t^{-k+1}\kappa  (x_1+\dots+x_d)^{k-2},~~~t\to\infty.
\end{equation}
Therefore, it follows from equations \eqref{3.12}, \eqref{4.15}, \eqref{4.31}, and \eqref{4.32} that
\begin{equation}\label{4.33}
\frac{1}{t^2}\text{Cov}\big[L_i,L_j\big|L(x)=t\big]= 0,~~~t\to\infty.
\end{equation}
As a result, from equations \eqref{4.23} and \eqref{4.33}, one obtains equation \eqref{4.28}. If we take $\varrho(x)=0$ in equation \eqref{3.15}, integrate it from $t(x)$ to $\infty$, and divide $\bar{F}_{L(x)}(t(x))$, then we have
\begin{equation}\label{4.34}
\mathbb{E}\big[L_iL_j\big|L(x)\geq t(x)\big]=\frac{1}{\bar{F}_{L(x)}(t(x))}\int_{t(x)}^\infty \frac{\partial }{\partial x_i}\mathbb{E}\big[L_j1_{\{L(x)\geq z\}}\big]dz.
\end{equation}
By substituting equation \eqref{4.24} into the integral in the above equation, one gets that 
\begin{equation}\label{4.35}
\int_{t(x)}^\infty \frac{\partial }{\partial x_i}\mathbb{E}\big[L_j1_{\{L(x)\geq z\}}\big]dz= t(x)^{-k+2}\frac{\kappa}{\kappa-2}(x_1+\dots+x_d)^{-\kappa+2},~~~t\to\infty.
\end{equation}
Consequently, if we substitute equations \eqref{4.14} and \eqref{4.35} into equation \eqref{4.34}, then we obtain equation \eqref{4.29}, where we take $t(x)=t$. That completes the proof.
\end{proof}

\subsection{$t$ Depends on $x$}

Now we suppose that the variable $t$ in previous Subsection 4.1 equals some risk measure $\varrho(x)$, which depends on the variable $x$ and converges to $\infty$. Simple examples of the risk measures, which satisfy the condition $\varrho(x)\to\infty$ are VaR$_\alpha(x)$ and ES$_\alpha(x)$, as $\alpha\to 1$. Under the assumption, equation \eqref{4.04} becomes
\begin{equation}\label{4.36}
\bar{F}_{L(x)}(\varrho(x))=\varrho(x)^{-\kappa}\ell_{L(x)}(\varrho(x)), ~~~\varrho(x)\to\infty
\end{equation}
for $\kappa>0$, where $\ell_{L(x)}(t)$ is slowly varying function at $\infty$. Since the density function $f_{L(x)}(t)$ is the ultimately monotone function, by the monotone density theorem, the density function $f_{L(x)}(\varrho(x))$ is written by 
\begin{equation}\label{4.37}
f_{L(x)}(\varrho(x))=\kappa \varrho(x)^{-k-1}\ell_{L(x)}(\varrho(x)),~~~\varrho(x)\to\infty.
\end{equation}

Using the ideas in proof of Proposition \ref{p.04}, we prove the following Proposition, which is an analog of Proposition \ref{p.04}.
\begin{prop}\label{p.06}
Let tail probability of a random variable $L(x)$ is represented by $\bar{F}_{L(x)}(\varrho(x))=\varrho(x)^{-\kappa}\ell_{L(x)}(\varrho(x))$ for $\kappa>1$, where $\ell_{L(x)}(\varrho(x))$ is slowly varying function at $\infty$. Then, we have that
\begin{equation}\label{4.39}
\lim_{\varrho(x)\to\infty}\frac{\mathbb{E}\big[L_i\big|L(x)=\lambda \varrho(x)\big]}{\mathbb{E}\big[L_i\big|L(x)=\varrho(x)\big]}=\lambda.
\end{equation}
\end{prop}

\begin{proof}
Thanks to equation \eqref{3.02}, we have that
\begin{equation}\label{4.40}
\frac{\mathbb{E}\big[L_i\big|L(x)=\lambda \varrho(x)\big]}{\mathbb{E}\big[L_i\big|L(x)=\varrho(x)\big]}=\frac{\lambda\frac{\partial}{\partial x_i}\varrho(x)+\frac{1}{f_{L(x)}(\lambda t)}\frac{\partial }{\partial x}\bar{F}_{L(x)}(\lambda t)}{\frac{\partial}{\partial x_i}\varrho(x)+\frac{1}{f_{L(x)}(t)}\frac{\partial }{\partial x}\bar{F}_{L(x)}(t)}.
\end{equation}
By substituting equations \eqref{4.04} and \eqref{4.06} into the above equation, one gets that
\begin{equation}\label{ad.02}
\frac{\mathbb{E}\big[L_i\big|L(x)=\lambda \varrho(x)\big]}{\mathbb{E}\big[L_i\big|L(x)=\varrho(x)\big]}=\lambda\frac{\frac{\partial}{\partial x_i}\varrho(x)+\frac{\varrho(x)^{\kappa+1}}{\kappa\ell_{L(x)}(\lambda \varrho(x))}\frac{\partial }{\partial x}\big\{\varrho(x)^{-\kappa}\ell_{L(x)}(\lambda \varrho(x))\big\}}{\frac{\partial}{\partial x_i}\varrho(x)+\frac{\varrho(x)^{\kappa+1}}{\kappa\ell_{L(x)}( \varrho(x))}\frac{\partial }{\partial x}\big\{\varrho(x)^{-\kappa}\ell_{L(x)}(\varrho(x))\big\}},~~~\varrho(x)\to\infty.
\end{equation}
According to equation \eqref{4.05}, one obtains
\begin{equation}\label{ad.03}
\lim_{\varrho(x)\to\infty}\frac{\mathbb{E}\big[L_i\big|L(x)=\lambda \varrho(x)\big]}{\mathbb{E}\big[L_i\big|L(x)=\varrho(x)\big]}=\lambda.
\end{equation}
Consequently, one gets equation \eqref{4.39}.
\end{proof}

\begin{prop}\label{p.07}
As $\varrho(x)\to\infty$, let tail probability of the portfolio random variable $L(x)$ is represented by $\bar{F}_{L(x)}(\varrho(x))=\varrho(x)^{-\kappa}\ell_{L(x)}(\varrho(x))$ for $\kappa>1$, for $i=1,\dots,d$, conditional expectations $\mathbb{E}\big[L_i\big|L(x)=\varrho(x)\big]$ are represented by 
\begin{equation}\label{4.41}
\mathbb{E}\big[L_i\big|L(x)=\varrho(x)\big]=\varrho(x)\ell_{L(x)}^i(\varrho(x)),
\end{equation}
and density $f_{L(x)}(t)$ be an ultimately monotone function, where $\ell_{L(x)}(t)$ and $\ell_{L(x)}^i(t)$ for $i=1,\dots,d$ are slowly varying function at $\infty$. Then, for the slowly varying functions, the following results hold
\begin{equation}\label{4.42}
\ell_{L(x)}^i(\varrho(x))=\frac{1}{x_1+\dots+x_d},~~~\varrho(x)\to\infty
\end{equation}
for $i=1,\dots,d$ and
\begin{equation}\label{4.43}
\ell_{L(x)}(\varrho(x))=(x_1+\dots+x_d)^\kappa,~~~\varrho(x)\to\infty.
\end{equation}
Consequently, we have that
\begin{equation}\label{4.44}
\bar{F}_{L(x)}(\varrho(x))=\varrho(x)^{-\kappa}(x_1+\dots+x_d)^\kappa,~~~\varrho(x)\to\infty
\end{equation}
and
\begin{equation}\label{4.45}
f_{L(x)}(\varrho(x))=\varrho(x)^{-\kappa-1}\kappa(x_1+\dots+x_d)^\kappa,~~~\varrho(x)\to\infty.
\end{equation}
\end{prop}

\begin{proof}
By multiplying equation conditional expectation \eqref{4.41} by $x_i$, summing it for $i=1,\dots,d$, and using the fact that $\mathbb{E}\big[L(x)\big|L(x)=\varrho(x)\big]=\varrho(x)$, we find that
\begin{equation}\label{4.46}
\sum_{i=1}^d\ell_{H(x)}^i(\varrho(x))x_i= 1.
\end{equation}
By equation \eqref{2.02} for $t=0$, we get that
\begin{equation}\label{4.47}
\frac{\partial \bar{F}_{L(x)}(\varrho(x))}{\partial x_i}=\bigg(\mathbb{E}\big[L_i\big|L(x)=\varrho(x)\big]-\frac{\partial \varrho(x)}{\partial x_i}\bigg) f_{L(x)}(\varrho(x))
\end{equation}
for $i=1,\dots,d$. If we substitute the tail probability \eqref{4.36}, conditional expectation \eqref{4.41}, and density \eqref{4.37} into equation \eqref{4.47}, then we obtain the same equation as equation \eqref{4.19}, namely,
\begin{equation}\label{4.48}
\frac{\partial \ell_{H(x)}(\varrho(x))}{\partial x_i}=\kappa\ell_{H(x)}^i(\varrho(x))\ell_{H(x)}(\varrho(x))
\end{equation} 
for $i=1,\dots,d$. By taking into account the constraint equation \eqref{4.46}, one gets a system of equations as equation \eqref{4.20}. As a result, we obtain equations \eqref{4.42} and \eqref{4.43}. From the representations \eqref{4.36} and \eqref{4.37} and equation \eqref{4.43}, one gets equations \eqref{4.44} and \eqref{4.45}.
\end{proof}

The following Corollary is a direct consequence of Proposition \ref{p.07}.
\begin{cor}\label{c.06}
Let conditions of Proposition \ref{p.07} hold. Then, the following results hold
\begin{equation}\label{4.49}
\mathbb{E}\big[L_i\big|L(x)=\varrho(x)\big]=\frac{\varrho(x)}{x_1+\dots+x_d},~~~\varrho(x)\to\infty
\end{equation}
and
\begin{equation}\label{4.50}
\mathbb{E}\big[L_i\big|L(x)\geq \varrho(x)\big]=\frac{\kappa}{\kappa-1}\frac{\varrho(x)}{x_1+\dots+x_d},~~~\varrho(x)\to\infty
\end{equation}
for $i=1,\dots,d$.
\end{cor}

\begin{proof}
Equation \eqref{4.49} follows from equations \eqref{4.41} and \eqref{4.42}. If we take $t(x)=\varrho(x)$ in equation \eqref{4.27}, then we get equation \eqref{4.50}.
\end{proof}

It will be interesting to consider ratios between the risk measures ES$_{\alpha}(x)$ and VaR$_{\alpha}(x)$ and their partial derivatives as quantile probability $\alpha$ converges to 1. By taking $\varrho(x)=q_\alpha(x)$ in Corollary \ref{c.06}, we get that
\begin{equation}\label{4.51}
\mathbb{E}\big[L_i\big|L(x)=q_\alpha(x)\big]=\frac{q_\alpha(x)}{x_1+\dots+x_d},~~~\alpha\to 1
\end{equation}
and
\begin{equation}\label{4.52}
\mathbb{E}\big[L_i\big|L(x)\geq q_\alpha(x)\big]=\frac{\kappa}{\kappa-1}\frac{q_\alpha(x)}{x_1+\dots+x_d},~~~\alpha\to 1
\end{equation}
for $i=1,\dots,d$. Since $\frac{\partial}{\partial x_i}\mathrm{VaR}_\alpha(x)=\frac{\partial}{\partial x_i}q_\alpha(x)=\mathbb{E}\big[L_i\big|L(x)=q_\alpha(x)\big]$ and $\frac{\partial}{\partial x_i}\mathrm{ES}_\alpha(x)=\mathbb{E}\big[L_i\big|L(x)\geq q_\alpha(x)\big]$, we have that
\begin{equation}\label{4.53}
\lim_{\alpha\to 1}\frac{\frac{\partial}{\partial x_i}\mathrm{ES}_\alpha(x)}{\frac{\partial}{\partial x_i}\mathrm{VaR}_\alpha(x)}=\frac{\kappa}{\kappa-1},~~~i=1,\dots,d.
\end{equation}
It also follows from Corollary \ref{c.06} that
\begin{equation}\label{4.54}
\frac{\partial}{\partial x_i}\mathrm{ES}_\alpha(x)=\frac{\kappa}{\kappa-1}{\frac{\partial}{\partial x_i}\mathrm{VaR}_\alpha(x)},~~~\alpha\to 1.
\end{equation}
If we multiply it $x_i$ and sum it for $i=1,\dots,d$, we obtain that
\begin{equation}\label{4.55}
\lim_{\alpha\to 1}\frac{\mathrm{ES}_\alpha(x)}{\mathrm{VaR}_\alpha(x)}=\frac{\kappa}{\kappa-1}
\end{equation}
for $\kappa>1$. Note that \citeA{McNeil15} show that the last equation holds for generalized Pareto distribution. Also, one can show that the formula holds for the student $t$ distribution. It is worth mentioning that for any risk measure $\varrho(x)$, which satisfies the condition $\varrho(x)\to\infty$, an exactly as same results as equation \eqref{4.53} and \eqref{4.55} hold, i.e., for $\kappa>1$,
\begin{equation}\label{4.56}
\lim_{\varrho(x)\to\infty} \frac{\mathbb{E}\big[L_i\big|L(x)\geq \varrho(x)\big]}{\mathbb{E}\big[L_i\big|L(x)= \varrho(x)\big]}=\frac{\kappa}{\kappa-1}
\end{equation}
and
\begin{equation}\label{4.57}
\lim_{\varrho(x)\to\infty} \frac{\mathbb{E}\big[L(x)\big|L(x)\geq \varrho(x)\big]}{\mathbb{E}\big[L(x)\big|L(x)= \varrho(x)\big]}=\frac{\kappa}{\kappa-1}.
\end{equation}

Now we give some results, which are related to asymptotic conditional covariances between components of the jointly continuous random vector $(L_1,\dots,L_d)'$.

\begin{cor}\label{c.07}
Let $\kappa>2$ and conditions of Proposition \ref{p.07} hold. Then, the following results hold
\begin{equation}\label{4.58}
\mathbb{E}\big[L_iL_j\big|L(x)=q_\alpha(x)\big]= \frac{q_\alpha^2(x)}{(x_1+\dots+x_d)^2},~~~\alpha\to 1
\end{equation}
and
\begin{equation}\label{4.59}
\mathbb{E}\big[L_iL_j\big|L(x)\geq q_\alpha(x)\big]= \frac{\kappa}{\kappa-2}\frac{q_\alpha^2(x)}{(x_1+\dots+x_d)^2},~~~\alpha\to 1
\end{equation}
for $i,j=1,\dots,d$.
\end{cor}

\begin{proof}
If we take $\varrho(x)=q_\alpha(x)$ in equation \eqref{3.12}, then since $\frac{\partial}{\partial x_i}\bar{F}_{L(x)}(q_\alpha(x))=0$, we have that
\begin{equation}\label{4.60}
\text{Cov}\big[L_i,L_j\big|L(x)=q_\alpha(x)\big]=\frac{1}{f_{L(x)}(q_\alpha(x))}\frac{\partial}{\partial x_i}\mathbb{E}\big[L_j1_{\{L(x)\geq q_\alpha(x)\}}\big].
\end{equation}
Due to equations \eqref{4.30} (whose $t$ is replaced $\varrho(x)$), \eqref{4.44}, and \eqref{4.50}, the partial derivative in the above equation is represented by
\begin{equation}\label{4.61}
\frac{\partial}{\partial x_i}\mathbb{E}\big[L_j1_{\{L(x)\geq q_\alpha(x)\}}\big]= -q_\alpha^{-\kappa}\frac{\partial q_\alpha(x)}{\partial x_i}\kappa (x_1+\dots+x_d)^{\kappa-1}+q_\alpha^{-\kappa+1}\kappa (x_1+\dots+x_d)^{\kappa-2}.
\end{equation}
Therefore, by equations \eqref{4.37}, \eqref{4.51}, \eqref{4.59}, and \eqref{4.61}, $\frac{1}{q_\alpha^2(x)}\text{Cov}\big[L_i,L_j\big|L(x)=q_\alpha(x)\big]= 0$, as $\alpha\to 1$. Consequently, equation \eqref{4.51} implies equation \eqref{4.58}. On the other hand, by taking $t=q_\alpha(x)$ in equation \eqref{4.34} and taking account equations \eqref{4.35} and \eqref{4.36}, we obtain equation \eqref{4.59}.
\end{proof}

From the above Corollary, one can easily conclude that for $\kappa>2$ and $i,j=1,\dots,d$,
\begin{equation}\label{4.62}
\lim_{\alpha\to 1}\frac{\mathbb{E}\big[L_iL_j\big|L(x)\geq q_\alpha(x)\big]}{\mathbb{E}\big[L_iL_j\big|L(x)= q_\alpha(x)\big]}= \frac{\kappa}{\kappa-2}.
\end{equation}
If we multiply equations \eqref{4.58} and \eqref{4.59} by $x_i$ and $x_j$ and sum it for all $i,j=1,\dots,d$, then for the second order conditional moments of the portfolio loss random variable, we obtain that
\begin{equation}\label{4.63}
\lim_{\alpha\to 1}\frac{\mathbb{E}\big[L^2(x)\big|L(x)\geq q_\alpha(x)\big]}{\mathbb{E}\big[L^2(x)\big|L(x)= q_\alpha(x)\big]}= \frac{\kappa}{\kappa-2},~~~\kappa>2.
\end{equation}
Following the ideas in the paper, for higher order conditional moments of the portfolio loss random variable, one may obtain similar formulas as equations \eqref{4.57} and \eqref{4.63}. Also, it follows from Corollary \ref{c.07} and equation \eqref{4.50} that for $\kappa>2$ and $i,j=1,\dots,d$, asymptotic conditional covariances of the loss random variables are given by
\begin{equation}\label{4.64}
\text{Cov}\big[L_i,L_j\big|L(x)\geq q_\alpha(x)\big]= \frac{\kappa}{(\kappa-2)(\kappa-1)^2}\frac{q_\alpha^2(x)}{(x_1+\dots+x_d)^2},~~~\alpha\to 1.
\end{equation}
As a result, for $\kappa>2$ and $i,j=1,\dots,d$, limits of conditional correlations of the loss random variables are given by
\begin{equation}\label{4.65}
\lim_{\alpha\to 1}\text{Corr}\big[L_i,L_j\big|L(x)\geq q_\alpha(x)\big]= 1.
\end{equation}

\section{Conclusion}

Here, we summarize the main findings of the paper: 
\begin{itemize}
\item[($i$)] For discrete portfolio loss random variable, we obtain that 
\begin{itemize}
\item the first order derivatives of any risk measure equal conditional expectations for given equality conditions, see equation \eqref{2.15}
\item and higher order derivatives of the risk measure are equal to zero, see equation \eqref{2.16}.
\end{itemize}
\item[($ii$)] For absolute continuous portfolio loss random variable, we introduce that
\begin{itemize}
\item three representation of the first order derivatives of any risk measure, see equations \eqref{3.02}, \eqref{3.03}, and \eqref{3.11},
\item a representation of positive homogeneous risk measure, see equation \eqref{3.08},
\item and the second order derivatives of any risk measure, see equation \eqref{3.14}.
\end{itemize}
\item[($iii$)] For heavy--tailed portfolio loss random variable, we obtain that
\begin{itemize}
\item representations of tail probability and density function of the portfolio loss random variable, see equations \eqref{4.14}, \eqref{4.15}, \eqref{4.44}, and \eqref{4.45},
\item the first order moments for given equality and inequality condition, see equations \eqref{4.23}, \eqref{4.24}, \eqref{4.49}, and \eqref{4.50},
\item the asymptotic ratio of ES to VaR also holds for any risk measure, which satisfies $\varrho(x)\to\infty$, see equation \eqref{4.57},
\item the second order moments for given equality and inequality condition, see equations \eqref{4.28}, \eqref{4.29}, \eqref{4.58}, and \eqref{4.59},
\item asymptotic ratio of the second order conditional moments of the portfolio loss random variable, see equation \eqref{4.63},
\item and limits of conditional correlations of loss random variables equal to 1, see equation \eqref{4.65}.
\end{itemize}
\end{itemize}
It is worth mentioning that the formulas in the papers are very simple. 

\section*{Acknowledgments}
We thank to National University of Mongolia for a grant that supported this study. This research was funded by National University of Mongolia under grant agreement (P2023-4568).

\bibliographystyle{apacite}
\bibliography{References}

\end{document}